\documentclass[a4paper,UKenglish,cleveref, autoref]{lipics-v2019}

\usepackage[utf8]{inputenc}
\usepackage{color} 
\usepackage{graphicx} 

\usepackage{booktabs} 
\usepackage{array} 
\usepackage{paralist} 
\usepackage{verbatim} 
\usepackage{amsmath}
\usepackage{amsfonts}   
\usepackage{amssymb}    
\usepackage[linesnumbered,ruled,vlined]{algorithm2e}

\newtheorem{assumption}{Assumption}

\usepackage[titles,subfigure]{tocloft} 

\newcommand\bigforall{\mbox{\Large $\mathsurround0pt\forall$}}
\newcommand\bigexists{\mbox{\Large $\mathsurround0pt\exists$}}

\title{On Tractability of Ulam's Metric in Higher Dimensions and 
Dually Related Hierarchies}

\author{Sebastian Bala}{Institute of Computer Science University of Opole, Kopernika 11a, Opole, Poland}{sbala@uni.opole.pl}{}{}

\author{Andrzej Kozik}{Institute of Computer Science University of Opole, Kopernika 11a, Opole, Poland}{akozik@uni.opole.pl}{}{}

\titlerunning{On Tractability of Ulam's Metric in Higher Dimensions}

\authorrunning{S. Bala and A. Kozik}



\nolinenumbers 

\hideLIPIcs  

\begin{document}
\maketitle
\begin{abstract} 
	Ulams’s metric defines the minimal number of arbitrary extractions and insertion of permutation elements and 
	to get the second permutation. The remaining elements constitutes the longest common subsequence of permutations.
	In this paper we extend Ulam’s metric $n$-dimensions. One dimension object is defined as a pair of permutations
	while $n$-dimension object is defined as a pair of $n$-tuples of permutations.
	For the purpose of encoding we explore $n$-tuple of permutations in order to define duallity of mutually related hierarchies.
	Our very first motivation comes from Murata, Fujiyoshi, Nakatake, and Kajitani paper, in which pairs of permutations are used as a representation of topological
	relation between an rectangles of a given size that can be placed within the area of minimal size.
	It is applicable to Very Large Scale Integration VLSI design. Our results concern hardness, approximability, and parameterized complexity within these hierarchies respectively.
\end{abstract}

\section{Introduction}

Permutation $\sigma \in S_n$ is a sequence $\left( \sigma(1), \sigma(2), \ldots, \sigma(n)\right)$ representing an arrangement of elements of the set $[n]$. 
Permutations are of great practical importance, as they model solutions to many real-life problems, e.g., in the fields of scheduling~\cite{Pinedo} and routing~\cite{VRP}. 
Taking into account significance and applications of permutations, a vast body of research tackled them from different angles, sides and points of view. 

One special research direction has been raised by Stanis\l{}aw Ulam questions~\cite{Ulam}: what is the minimum number of moves to go from some permutation $\sigma_s \in S_n$ to some other permutation $\sigma_t \in S_n$, 
where a move consists of changing a position of some element in a permutation, and more importantly, what is asymptotic distribution of this number? 

The answer for the first question is named \emph{the Ulam's metric} $U(\sigma_s, \sigma_t)$. Computationally, it is equivalent to $LCS(\sigma_s, \sigma_t)$ - \emph{the Longest Common Subsequence} of $\sigma_s$ and $\sigma_t$, and 
Fredman~\cite{Fredman} showed that it can be computed using $n \log n - n \log \log n + O(n)$ comparisons in the worst case, and no algorithm has better worst-case performance. Later, Hunt and Szymanski~\cite{Hunt} improved this bound to 
$O(n \log \log n)$ by exploiting RAM model of computations and a fixed universe of sequence elements.

The latter Ulam's question drew researcher's attention over the past 50 years, providing many amazing results and connections between various areas of mathematics, physics and statistics, exhibited and reviewed 
in~\cite{Aldous99longestincreasing},\cite{Romik},\cite{Corvin}, and foremost \emph{the final answer} of Baik, Deift and Johansson~\cite{Baik}: the expected value of the length of $LCS(\sigma_s, \sigma_t)$ for $\sigma_s$ and $\sigma_t$ drawn uniformly at 
random from $S_n$ equals $2\sqrt{n} - 1.77108n^{1/6} + o(n^{1/6})$.

As permutations often express solutions to many combinatorial optimization problems, the Ulam's metric is often interpreted as the length of the shortest path connecting (any) two given solutions in a space spanned by $S_n$.
 Motivated by surprising discoveries that followed original Ulam's questions \cite{Corvin},\cite{Romik}, in this paper we rephrase and state the same question, but in higher dimensions: what is the length of the shortest path 
 between two solutions in $S_n^k$,
\begin{equation*}
S_n^k = \underbrace{S_n \times \ldots \times S_n}_{k\mathrm{~times}},
\end{equation*}
the $k$-dimensional space of permutations?

The question is not of theoretical interest only, as 2-tuples of permutations describe non-overlapping packings of rectangles on a~plane in \emph{Sequence Pair} (SP) representation \cite{Murata}. The SP has many successful industrial applications in 
the context of physical layout synthesis of VLSI circuits, where it models a~placement of transistors, leaf-cells and macro-blocks on a silicon die~\cite{Markov}. On the other hand, SP can be a solution space for multiprocessor scheduling problems with 
various definitions of cost functions and constraints~\cite{Imahori}\cite{Kozik17}\cite{KozikCOR}. 

On the other hand, existence and properties of paths between solutions in a solution space are especially important for metaheuristic algorithms~\cite{Blum}, performing effective exploration of the solution space based on elimination and 
explorative properties - incremental moves and neighbourhood structures are essential in the context of guiding the search process. Examples of such methods are hybrid-metaheuristics~\cite{Blum2}, combining complementary strengths of various techniques 
to collaboratively tackle hard optimization problems, and hyper-heuristic approaches~\cite{Burke}, providing a robust upper-level framework to tune and drive underlying heuristics to adapt to features of solved problems. The efficiency of such methods 
is based on the connectivity and diameter properties of the solution space - the assumption that one can provide a~sequence of moves to every (starting) solution such that every other solution can be reached (especially an optimal one), and the maximal 
length of such a sequence, respectively. On the other hand, algorithms behind the Ulam's metric can be immediately applied as crossover operators in evolutionary and path-relinking metaheuristics~\cite{Blum}., i.e., given a shortest path between two solutions, 
a result of their crossover could be either a midpoint, the best, or even all solutions on that path as an offspring.

\section{Preliminaries}
Let a subsequence $s$ of $\sigma \in S_n$ be a sequence $\left( \sigma(i_1), \sigma(i_2), \ldots, \sigma(i_m)\right)$ where $1 \leq i_1 < \ldots < i_m \leq n$; denote by $\{s\}$ a set of elements of $s$, i.e., 
$\{s\} = \{ \sigma(i_1), \sigma(i_2), \ldots, \sigma(i_m) \} \subseteq [n]$ and by $l(s) = m$ its length.
A set common subsequences of two permutations we denote by
\begin{eqnarray*}
CS(\sigma_s, \sigma_t) & = & \left\{\sigma \in [n]^m \;:\; m \in [n] \hbox{ and } \forall {i,j \in [m]}\right.  \\
& & \left. i<j \Rightarrow  (\sigma^{-1}_s(\sigma(i)) < \sigma^{-1}_s(\sigma(j)) \;\wedge\;
\sigma^{-1}_t(\sigma(i))<\sigma^{-1}_t(\sigma(j)) )\right\}.
\end{eqnarray*}

 Let $\Gamma \in S_n^k$ be a $k$-tuple of permutations of $[n]$, i.e., $\Gamma = (\sigma^1, \ldots, \sigma^k)$, $\sigma^i \in S_n$, $i \in [k]$. Say an \emph{insert~move} of some $v \in [n]$ in $\Gamma$ consist of moving the element $v$ to some other 
position in each of $\Gamma$ permutations. Define an insert neighbourhood of $\Gamma$, $\mathcal{N}(\Gamma)$, as a set of permutation tuples obtained by performing single insert move of any element in $\Gamma$.

Let $f_k: S^k_n \times S^k_n \times 2^{[n]} \rightarrow [n]$ be defined in the following way:
$$
f_k(\Gamma_s,\Gamma_t,C) = \left\{
\begin{array}{ll}
|C| & \mbox{ if } \mathop{\bigforall} \limits_{i \in [k]} \; \mathop{\bigexists} \limits_{\sigma \in CS(\sigma_s^i, \sigma_t^i)} \{\sigma\} = C \\
0 & \mbox{ otherwise }
\end{array}
\right.
$$
\begin{eqnarray} \label{Uk}
U_k(\Gamma_s, \Gamma_t) & = &\{ [n] \setminus C\;:\; C\in LFS(\Gamma_s, \Gamma_t) \}, \\
& & \mbox{ where } LFS(\Gamma_s, \Gamma_t) = \operatorname*{argmax}_{C \subseteq [n]} \; f_k(\Gamma_s,\Gamma_t,C) \nonumber
\end{eqnarray}
Note that $LFS(\Gamma_s, \Gamma_t)$, as well as $U_k(\Gamma_s, \Gamma_t),$ is a class of equsized sets.

Define a path between $\Gamma_s$ and $\Gamma_t$ in $S_n^k$ as a sequence $\Gamma_0 = \Gamma_s, \Gamma_1, \ldots, \Gamma_{m-1}, \Gamma_m = \Gamma_t$ s.t. $\Gamma_{i+1} \in \mathcal{N}(\Gamma_i)$, $i = 0, \ldots, m-1$; let $m$ be the length of that path. 
The Ulam's metric in $S_n^k$, is the length of the shortest path between $\Gamma_s$ and $\Gamma_t$, i.e., the minimal number of insert moves transforming $\Gamma_s$ into $\Gamma_t$, where $\Gamma_s, \Gamma_t \in S_n^k$.
Following \cite{Aldous99longestincreasing} it is not hard to show that the Ulam's metric is the size of sets in $U_k(\Gamma_s, \Gamma_t)$,

A \textit{sequence pair} $(\sigma^{1},\sigma^{2}) \in S_n^2$ is a pair of permutations over $[n]$. 
In the literature, according to Ulam's original motivation, $U_1$ is often interpreted as sorting of a hand of bridge cards, by sequentially picking cards
labelled by numbers which are not in the longest common sequence and putting them back in desired position. It is shown in Figure~1, using sequence pairs representation, 
the $U_2$ can be analogously interpreted geometrically as a sequence of rectangle packings, in which in each step a single rectangle is taken from a packing and re-inserted in desired place. 
To more precise explanation how sequence pairs corresponds to rectangle placement 
see \cite{Murata}. 
\begin{figure}
\includegraphics[angle=0, width=11.5cm, height=3cm]{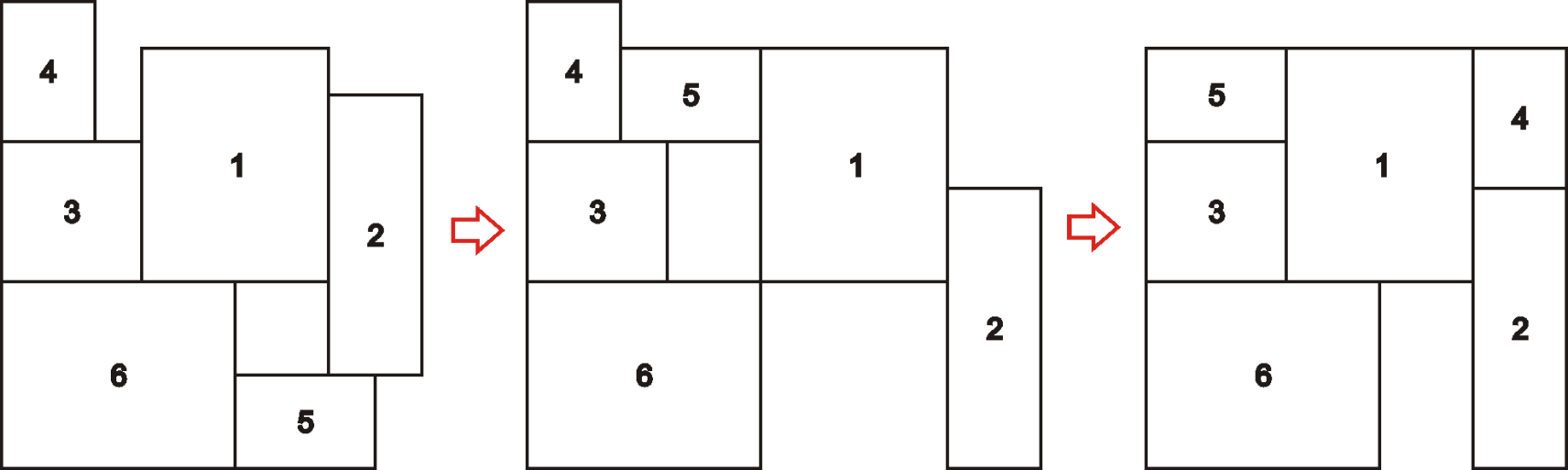}
\caption{$ \langle (43162\mathbf{5}),(6\mathbf{5}3412) \rangle \rightarrow \langle (\mathbf{4}53162),(63\mathbf{4}512) \rangle \rightarrow \langle (531462),(635124) \rangle$}
\label{recrec}
\end{figure}

For two given finite sequences $a$ and $b,$ its join is denoted by $a\cdot b.$
By $(i),$ we denote a sequence that consists of a single element $i.$
For a given pair $(\sigma^i_s,\sigma^i_t),$ a subset $B\subseteq [n]$ \textit{induces} a common subsequence
if there is a sequence $(b_1,\ldots, b_l) \in CS(\sigma^i_s,\sigma^i_t)$ such that the $\{b_1,\ldots, b_l\}=B.$
For given $B,$ we also say, that a sequence $(b_1,\ldots, b_l)$ is a $B$-{\it induced} subsequence of permutation $\sigma$
if $(b_1,\ldots, b_l)$ is a subsequence of $\sigma$ and $\{b_1,\ldots, b_l\}= B.$

Probably the most convenient form of defining of an \textsc{NPO} problem $A$ is its
presentation as a fortuple $(\mathcal{I}(A),fso,cost_A,type),$ where:
\begin{itemize}
	\item $\mathcal{I}(A)$ is a set of valid instances of $A.$
	\item $fso$ is a function such that for given an instance $x\in \mathcal{I}(A),$ $fso(x)$ is a set of feasible solutions for $x.$
	Moreover the question if $y\in fso(x)$? is verifiable in polynomial time with respect to the size of $x.$
	\item Given an $x\in\mathcal{I}(A)$ and a feasible solution $y$ of $x,$ $cost_A(x,y)$ is positive integer measure of $y.$
	Additionally $cost_A(x,y)$ is computable in polynomial time.
	\item $type \in \{\min,\max\}.$
\end{itemize}

Given a pair of $k$-tuple of permutations $\Gamma=\langle \Gamma_s, \Gamma_t\rangle,$ where $\Gamma_s=(\sigma^1_s,\ldots \sigma^k_s)$ and  
$\Gamma_t=(\sigma^1_t,\ldots \sigma^k_t),$ a set $B \subseteq [n]$ is a {\it feasible solution} of $\Gamma$ if $B$
induces a common subsequence for every $(\sigma_s^{i},\sigma_t^{i})$, where $i=1,\ldots,k.$
The $k$-\emph{dimensional Largest Fixed Subset Problem} (k\textsc{LFS}) is defined as follows:
Given a pair of $k$-tuple of permutations $\Gamma=\langle \Gamma_s, \Gamma_t\rangle,$ where $\Gamma_s=(\sigma^1_s,\ldots \sigma^k_s)$ and  
$\Gamma_t=(\sigma^1_t,\ldots \sigma^k_t),$ determine a maximal $l$ and $B \subseteq [n]$ such that $|B|=l$ and $B$ is a feasible solution of $\Gamma.$
The cost function $cost_{k\textsc{LFS}}$ satisfies $cost_{k\textsc{LFS}}(\Gamma,B)=|B|.$

A decision version of $k$-\emph{dimensional Longest Fixed Subset Problem} (k\textsc{LFSD}) is answering the question if for given nonnegative integer $l$ 
and $\langle \Gamma_s, \Gamma_t\rangle= \langle (\sigma^1_s,\ldots \sigma^k_s), (\sigma^1_t,\ldots \sigma^k_t)\rangle$
there exists $B \subseteq [n],$ such that $|B|=l$ and $B$ is a feasible solution of $\Gamma.$

By $k\textsc{U}$ we denote a dual problem to k\textsc{LFS}. The $k\textsc{U}$ is defined over the same inputs as k\textsc{LFS}
and if for a given input, a set $B$ is a feasible solution to the k\textsc{LFS} then a set $V \setminus B$ is a feasible solution to the $k\textsc{U}$ problem. 
The $k\textsc{U}$ problem consists in find the smallest feasible solution. The $cost_{k\textsc{U}}$ satisfies $cost_{k\textsc{U}}(\Gamma,B)=|B|.$  
The decision version $k\textsc{UD}$ is defined analogously to the k\textsc{LFSD} problem.

In the \textsc{MaxClique} problem we are given undirected graph $G = (V, E)$ and the task is to find a clique of maximum size.
The decision version of the \textsc{MaxClique} problem -- the \textsc{Clique} problem  consists in checking for a given $m$ if the graph consists a clique of size $m.$
For the \textsc{MinVertexCover} problem the task is to find a subset $B\subseteq V$ of minimal cardinality such that, for every $(u,v) \in E,$ $u \in B$ or $v\in B.$  

Solutions of the \textsc{MaxClique} and the \textsc{MinVertexCover} problem are related to each other in the following way:
subset $B$ of $V$ is a minimum vertex cover if and only if $V \setminus B$ is a maximum clique
in $G.$ The \textsc{VertexCover} problem denote decision version of the \textsc{MinVertexCover} problem.

 Let $A$ be an maximization problem. We use $Opt_A(x)$ to refer to the cost of optimal solution for $x\in \mathcal{I}(A).$
An algorithm $\mathcal{A}$, for $A,$ computes objective function value $\mathcal{A}(x)$ for feasible solution of an instance $x.$
Define $F(\mathcal{A},x)=\frac{Opt_A(x)}{\mathcal{A}(x)}.$ Function $r:\mathbb{Z}^+ \mapsto \mathbb{R}^+$ is an \textit{approximation factor} of $\cal A$ if for any $n,$ $F(\mathcal{A},x) \leq r(n)$ 
for all instances of the size $n.$ 
We say that algorithm $\mathcal{A}$ is $r(n)$-approximation algorithm. If $A$ is an minimization problem then $F(\mathcal{A},x)=\frac{\mathcal{A}(x)}{Opt_A(x)}.$

Let $cost_A(x,y)$ denote the cost of $y$ that is feasible solution of $x\in\mathcal{I}(A).$
A \emph{reduction} of an optimization problem $A$ to an optimization problem $B$ is a pair of polynomially computable functions $(f,g)$
that satisfy two conditions: $(1)$ $x'=f(x) \in \mathcal{I}(B)$  for any $x\in \mathcal{I}(A),$ $(2)$ $y=g(x,y')$ is feasible solution of $x,$ for any $y'$ which is feasible solution of $x'.$
A reduction $(f,g)$ of an optimization problem $A$ to an optimization problem $B$ is an \emph{S-reduction} if $(1)$ $Opt_A(x)=Opt_B(x')$ 
for any $x\in \mathcal{I}(A),$ where $x'=f(x),$ and $(2)$ $cost_A(x,g(x,y'))=cost_B(x',y'),$ for any $x\in \mathcal{I}(A)$
and $y'$ which is feasible solution of $x'$  \cite{PierluigiGuide}.
For abbreviation, we write  $A \leq_S B$ to denote that there exists $S$-reduction from $A$ to $B.$

A decision problem $A$ is \emph{parametrized problem} if it is extended by function $\kappa_A$ that assigns nonnegative integer to each instance $x$ of $A.$ An algorithm which decides if $x\in A$ in time $f(\kappa_A(x))p(|x|)$ is called \emph{fpt-algorithm}, 
where $f$ is a computable function and 
$p$ is a polynomial. An $R$ is an \emph{fpt-reduction} of $A$ to $B$ if for any $x\in \mathcal{I}(A),$ $R(x)$ is an instance of $B,$
moreover $(1)$ $x\in A$ iff $R(x)\in B;$ $(2)$ $R$ is computable by an fpt-algorithm with respect to $\kappa_A;$  $(3)$ There is computable function $h$ such that $\kappa_B(R(x)) \leq h(\kappa(x))$ for any $x\in \mathcal{I}(A).$

\section{On Tractability of k\textsc{LFS} and k\textsc{U}}
\subsection{S-reductions and some conclusions}

\begin{lemma}
\label{firstL}
For any $k\in [n]$ there exist S-reductions from k\textsc{LFS} to \textsc{MaxClique} and from k\textsc{U} to \textsc{MinVertexCover}.
\end{lemma}
\begin{proof} For a given $\langle \Gamma_s , \Gamma_t \rangle = $ $\langle (\sigma^1_s,\ldots,\sigma^k_s), (\sigma^1_t,\ldots,\sigma^k_t) \rangle$ we define a 
graph $G=(V=[n],E),$ where $(i,j) \in E$ if and only if for all $r \in [k]$ there are 
$\alpha_s, \beta_s, \alpha_t, \beta_t \in [k]:$ 
$\alpha_s < \beta_s,$ $\alpha_t < \beta_t$  and either $\sigma^r_s(\alpha_s)=i,\sigma^r_s(\beta_s)=j, \sigma^r_t(\alpha_t)=i,\sigma^r_t(\beta_t)=j$ or $\sigma^r_s(\alpha_s)=j,\sigma^r_s(\beta_s)=i, \sigma^r_t(\alpha_t)=j,\sigma^r_t(\beta_t)=i.$ 
This defines the function $f$ transforming any given instance of k\textsc{LFS} instance of \textsc{MaxClique}.  

First, note that any clique $C$ of $G,$ $C$-induces a common sequences for all pairs $\langle \sigma^i_s, \sigma^i_t\rangle$. 
Suppose to the contrary, that $C$ does not induce a common sequence for some $(\sigma^i_s,\sigma^i_t)$. 
The clique $C$ has at least two elements. Let $\alpha_s,\alpha_t$ be  the $C$-induced subsequences of $\sigma^i_s$ and $\sigma^i_t, $ respectively.
If $\alpha_s, \alpha_t$ are two different sequences over the same set of elements differs then 
there is the least index $d$ such that $\alpha_s(d) \neq \alpha_t(d).$
Therefore the element $\alpha^{-1}_s(\alpha_t(d))>\alpha^{-1}_s(\alpha_s(d))$ and $\alpha^{-1}_t(\alpha_t(d))<\alpha^{-1}_t(\alpha_s(d))$. 
The reduction has been defined in such a way that there is no edge $(\alpha_s(d),\alpha_t(d)) \in E.$ 
This is contradiction - $C$ cannot form a clique in $G.$ Therefore $g(\langle \Gamma_s , \Gamma_t \rangle,C)=C.$ So far it has been proven that 
there is reduction from k\textsc{LFS} to \textsc{MaxClique}.
Note that $cost_{MaxClique}(f(\langle\Gamma_s,\Gamma_t\rangle)=G,C)=|C|$ for any clique $C$ of $G.$ Additionally,
$cost_{kLFS}(\langle\Gamma_s,\Gamma_t\rangle,g(\langle\Gamma_s,\Gamma_t\rangle,C))=C$ can be defined as $|C|$ for any $C$ which is a clique.
Hence, the second condition of the S-reduction definition may also be satisfied.

Finally, note that maximum size clique $C$ of the graph $G$, represented as the set of vertex numbers, induces a maximum solution of the k\textsc{LFS} instance. 
To prove this assume to the contrary that there exists $D \subseteq V$ such that there are $D$-induced subsequences for all pairs $(\sigma_s^i,\sigma_t^i)$
and $|C|<|D|$ for all cliques $C \subseteq V.$ Therefore there are different $i,j\in D$ such that $(i,j)\not \in E.$ 
By definition of the reduction there exists $r,$ such that either $\sigma^r_s(\alpha_s)=i,\sigma^r_s(\beta_s)=j, \sigma^r_t(\alpha_t)=j,\sigma^r_t(\beta_t)=i$ 
or $\sigma^r_s(\alpha_s)=j,\sigma^r_s(\beta_s)=i, \sigma^r_t(\alpha_t)=i,\sigma^r_t(\beta_t)=j,$ for some 
$\alpha_s < \beta_s,$ $\alpha_t < \beta_t.$ This contradicts the fact that $i,j$ are elements of a D-induced common sequence.

Now we see that the function $g:S^k_n \times S^k_n \times 2^{V} \mapsto 2^{V}$ satisfying $g((\Gamma_s,\Gamma_t),C)=C.$ 
also satisfies $Opt_{MaxClique}(G=f(\langle\Gamma_s,\Gamma_t\rangle))=Opt_{kLFS}(\langle\Gamma_s,\Gamma_t\rangle)$ for any 
$\langle\Gamma_s,\Gamma_t\rangle.$ Therefore, the first condition of the S-reduction definition is also satisfied.
Thus, the $(f,g)$ is an S-reduction from k\textsc{LFS} to \textsc{MaxClique}. 

A set $B\subseteq[n]$ is a feasible solution of k\textsc{U} problem for instance $(\Gamma_s,\Gamma_t)$ if and only if $[n] \setminus B$ is a feasible solution of the k\textsc{LFS} problem for the same instance $(\Gamma_s,\Gamma_t).$  
Similarly if $B\subseteq V=[n]$ is a vertex cover of $G$ then $V\setminus B$ is a
clique for $G.$ Moreover, the minimal vertex cover $B$ for $G$ corresponds to $V\setminus B$ which is a maximum size clique of $G.$ 
Similarly, a  maximal longest fixedsubset $B$ of $(\Gamma_s,\Gamma_t)$ corresponds to $[n]\setminus B$ which is a minimal solution of the $k\textsc{U}$ problem. 
The S-reduction $(f',g')$ from \textsc{U} to \textsc{MinVertexCover} is defined by equations $f'=f$ 
and $g'((\Gamma_s,\Gamma_t),B)=[n] \setminus g((\Gamma_s,\Gamma_t),V \setminus B).$ The equality of the costs can be proven using the same reasoning as before.
\end{proof}

An $S$-reduction is stronger than $L$-reduction, $AP$-reduction and $PTAS$-reduction \cite{PierluigiGuide}, that are often used for proving the membership to the 
$APX$ class. Furthermore, if $A \leq_S B$ and $B$ is approximable with some factor, then $A$ is approximable with the same factor.
Therefore, by Lemma \ref{firstL} and the fact that \textsc{MinVertexCover} there exists $2$-approximation algorithm \cite{Vazirani}, we obtain that 
there is $2$-approximation algorithm for the k\textsc{U} problem. 

For verification if an $Y$ is feasible solution for an $x\in \mathcal{I}(k\textsc{LFS})$ one can iteratively check if consecutive 
elements of $\sigma^i_s$ are in $Y.$ These iteration puts consecutive elements to the new sequence $\tau^i_s$ if only they are in $Y.$
The same can be done for $\sigma^i_t,$ obtainig $\tau^i_s$ for all $i\in [k].$ At the end we check if $\sigma^i_t = \tau^i_s$ for each $i\in [k].$
This can be implemented within $O(kn\log (|Y|)).$ Then fasible solutions are verifiable in polynomial time. Since $cost_{k\textsc{LFS}}(Y)=|Y|,$  
$cost_{k\textsc{LFS}}$ is obviously computable in polynomial time. 

As a conclusion we obtain:

\begin{theorem}
\label{secondT}
The $k$\textsc{LFS} problem is in $NPO$ for any $k\in [n]$.  For any $k\in [n]$ there exists 2-approximation algorithm for k\textsc{U}.
\end{theorem}

\begin{lemma}\label{TwoSReductions}

There exists an $S$-reduction from the \textsc{MaxClique} problem to the $n$\textsc{LFS} problem  and from the \textsc{MinVertexCover} problem  to $n$\textsc{U} problem.  
\end{lemma}
\begin{proof}
Let $G=(V,E)$ be undirected graph, where $V=[n].$ We define $\Gamma=\langle \Gamma_{s},\Gamma_{t} \rangle,$
where $\Gamma_{s}=(\sigma_s^1,\ldots,\sigma_s^n)$ and $\Gamma_{t}=(\sigma_t^1,\ldots,\sigma_t^n).$ Firstly define permutation $\sigma^{i}_{s}.$ Let $V^{-i}= V\setminus \{i\}.$ Split $V^{-i}$ into two disjoin sets 
$P^i=\{j\in V^{-i}\;|\;jEi\}$ and  $N^i=\{j\in V^{-i}\;|\; \neg jEi\}.$

A sequence $a$ is \textit{non repeating} if there is no element in $a$ occuring twice.
Let $a_i$ and $b_i$ be a increasing non repeating sequences of all elements from $N^i$ and $P^i,$ respectively. Define $\sigma^{i}_{s}$ as $(i)\cdot a_i \cdot b_i.$ 
The second permutation $\sigma^{i}_{t}$ is defined as $ a_i \cdot (i) \cdot b_i.$

\begin{remark}
	\label{RFour}
	There is a clique $K$ in the graph $G$ if and only if  $K$ is feasible solution of $\langle \Gamma_{s},\Gamma_{t}\rangle.$ 
\end{remark}	
\begin{remark}
\label{RFive}
If the size of maximal clique in graph $G$ equals $m,$ then $\langle \Gamma_{s},\Gamma_{t} \rangle$ has no feasible solution of length greater than $m.$
\end{remark}

The reduction is defined in the following way: $f(G)=\Gamma,$ $g(G,K)=K.$  By Remark \ref{RFour}, for given 
a feasible solution $K$ of the $n$\textsc{LFS} problem for instance $\Gamma,$ $g(G,K)$ is feasible solution of the \textsc{MaxClique} problem for instance $G.$ 
Then $g$ is correctly defined. By Remarks \ref{RFour} and \ref{RFive}, $Opt_{\textsc{MaxClique}}(G)=Opt_{n\textsc{LFS}}(\Gamma).$
For \textsc{MaxClique} as well as $n$\textsc{LFS} cost of feasible solution is the size of a set which constitutes the solution.
Hence $cost_{n\textsc{LFS}}(\Gamma,K)=cost_{\textsc{MaxClique}}(G,g(G,K))=|K|.$

Like in the proof of Lemma \ref{firstL} presented $S$-reduction is also a reduction from \textsc{MinVertexCover} to $n$\textsc{U}.
This is because the graph $G$ which is an instance of the \textsc{MinVertexCover} problem is also an instance of the \textsc{MaxClique} problem.
A maximum solution $B$ of $\Gamma$ which is an instance of the $n$\textsc{LFS} problem corresponds to $[n]\setminus B$ which is a minimum solution 
of the same $\Gamma$ which is an instance of the $n$\textsc{U} this time. 
\end{proof}

\subsection{Np-harness and nonapproximability}

\begin{theorem}
\label{firstT}
The 2\textsc{LFSD} and 2\textsc{UD} are \textsc{NP}-complete.
\end{theorem}
\begin{proof}

By the existence of an $S$-reduction (Lemma \ref{firstL}) from k\textsc{LFS} to \textsc{MaxClique} and from k\textsc{U} to \textsc{MinVertexCover}, 
we get that for any $k\in [n],$ the k\textsc{LFSD} and the k\textsc{UD} are reducible to the \textsc{Clique} 
problem and the k\textsc{UD} problem, respectively, by polynomial time reductions. Hence $k\textsc{LFS}$ and k\textsc{UD} are in $NP$  and in particular $2\textsc{LFS}$ and k\textsc{UD} are in $NP.$

We let remains to show that 2\textsc{LFSD} is in \textsc{NP}-hard. Let $\varphi=c_1\wedge \ldots \wedge c_m$ be a boolean formula in $3CNF$ form,
where $3CNF$ means that each clause $c_i$ has three literals. Let $x_1,\ldots, x_n$ be all variables that occur in $\varphi.$ 
Assume that variable $x_i$ occurrs $\alpha$ times in $\varphi$ as a positive literal and $\beta$ times as negative one. We construct polynomial time 
reduction by encoding a satisfiability of $\varphi$ as 2\textsc{LFSD}. We introduce a new set of symbols 
$x_i^{(1)},\ldots,x_i^{(\alpha)}, \neg x_i^{(1)},\ldots, \neg x_i^{(\beta)},$ where  $x_i^{(j)}$
stands for the $j$-th positive literal of $x_i$ and symbol $\neg x_i^{(j)}$ stands for the $j$-th negative literal of $x_i.$ 
Let $\mathcal{C}_i=\{x_i^{(1)},\ldots,x_i^{(\alpha)}, \neg x_i^{(1)},\ldots, \neg x_i^{(\beta)}\}$ and $\mathcal{C}=\bigcup_{i=1}^n \mathcal{C}_i.$
Let $\kappa$ be the bijection from $\mathcal{C}$ onto $[3m].$ Create $\Gamma(\varphi)=\langle \Gamma_s,\Gamma_t\rangle =$ $\langle (\sigma^1_s,\sigma^2_s),(\sigma^1_t,\sigma^2_t)\rangle$ 
which is an instance of $2$\textsc{LFSD} having solution $m$ if and only if $\varphi$ is satisfiable. 

The definition of $\Gamma(\varphi) = \langle \Gamma_s,\Gamma_t\rangle$ in the presented reduction should preclude the case
that if $Z$ induces common subsequences for $(\sigma^1_s,\sigma^1_t)$ and $(\sigma^2_s,\sigma^2_t)$
then both $\kappa(x_i^{(r)})$ and $\kappa(\neg x_i^{(p)})$ appear in $Z$ for any $i\in [n]$ and $r,p.$ It can be realized in the following way:
Permutations $\sigma_{s}^{1}$ and $\sigma_{t}^{1}$ consist of blocks of substrings $A_i=\left(\kappa(x_i^{(1)}),\ldots,\kappa(x_i^{(\alpha)})\right)$ and 
$B_i=\left(\kappa(\neg x_i^{(1)}),\ldots, \kappa(\neg x_i^{(\beta)})\right).$ We define $\sigma_{s}^1=$ $A_1B_1A_2B_2\cdots A_nB_n$ and
$\sigma_{t}^1=B_1A_1B_2A_2\cdots B_nA_n.$

\begin{remark} \label{pierwszapara}
If a set $Z$ induces a common subsequence of $\sigma_{s}^1$ and $\sigma_{t}^1$
then at most one of $\kappa(x_i^{(r)})$ and $\kappa(\neg x_i^{(p)})$ is in $Z$ because $(1)$ $\kappa(x_i^{(r)})$ appears before $\kappa(\neg x_i^{(p)})$ in $\sigma_{s}^1,$ $(2)$ $\kappa(x_i^{(r)})$ appears after $\kappa(\neg x_i^{(p)})$ in $,\sigma_{t}^1$
for any $i,r,p.$
\end{remark}

Let $c_j=x_{j1}^{(a)}\vee x_{j2}^{(b)} \vee x_{j3}^{(c)}$ be the $j$-th clause in $\varphi,$ where 
$x_{j1}^{(a)}, x_{j2}^{(b)}, x_{j3}^{(c)}$ are literals.
Define a permutation $\sigma^{2}_{s}$ as concatenation of blocks $E_1E_2\cdots E_m$ where $E_j = (\kappa(x_{j1}^{(a)}),\kappa(x_{j2}^{(b)}),\kappa(x_{j3}^{(c)}))$ and $\sigma^{2}_{t}=E^R_1E^R_2\cdots E^R_m,$ where $E^R_j = (\kappa(x_{j3}^{(c)}),\kappa(x_{j2}^{(b)}),\kappa(x_{j1}^{(a)})).$

\begin{remark} \label{drugapara}
If a set $Z$ induces a common subsequence of $\sigma^{2}_{s}$ and $\sigma^{2}_{t}$ then $|Z| \leq m.$
If $|Z|\geq m + 1$, there would exist literals $x,y$ belonging to the same clause such that $\kappa(x),\kappa(y) \in Z.$ 
Literals $\kappa(x),\kappa(y)$ that occur in $\sigma^{2}_{t},$  appears in reverted order in $\sigma^{2}_{t}.$ 
Hence $Z$ does not induce a common subsequence for both $\sigma^{2}_{s}$ and $\sigma^{2}_{t}.$
\end{remark}
Any choice of literals, one from every clause forms the sequence, which is \textit{witness of satisfiability} of $\varphi$ if the set of chosen literals is consistent.

\begin{remark}\label{trzeciapara}
Assume that $\varphi$ is satisfiable and the set $\{y_1,\ldots, y_m\}$ is a witness of its satisfiability.
Then the sequence $(\kappa(y_1),\kappa(y_2),\cdots,\kappa(y_m))$ is a common subsequence of $\sigma^{2}_{s}$ and $\sigma^{2}_{t}.$
\end{remark}

Note that if $y_1,\ldots, y_m$ is a witness of satisfiability then $\{\kappa(y_1),\kappa(y_2),\ldots \kappa(y_m)\}$ indicates a common subsequence of $\sigma^{1}_{s}$ and $\sigma^{1}_t.$
Indeed, by consistency of $\{y_1,\ldots,$ $ y_m\}$ there is no $i,j,k$ such that $\kappa(y_i)$ occurs in $A_k$ and $\kappa(y_j)$ occurs in $B_k.$ 
Hence, $\kappa(y_i)$ and $\kappa(y_j)$ occur in the same order in both $\sigma^{1}_{s}$ and $\sigma^{1}_{t}.$ Therefore, by Remark \ref{trzeciapara}, 
$\{\kappa(y_1),\kappa(y_2),\ldots \kappa(y_m)\}$ induces common subsequences for pairs $(\sigma^{1}_{s},\sigma^{1}_{t})$ and $(\sigma^{2}_{s},\sigma^{2}_{t})$ of length $m.$ By Remarks \ref{pierwszapara} and \ref{drugapara} the lengths of induced subsequences are not greater than $m.$
Thus, the induced subsequences are of the maximal length.

By the above reduction $\varphi$ is satisfiable iff the solution of $\langle \Gamma_s,\Gamma_t\rangle,$ as the instance of the the 2\textsc{UD} problem 
is of the size not greater than $2m.$ Hence 2\textsc{UD} is $NP$-hard.
\end{proof}
For a given boolean formula $\varphi$ in $3CNF,$  the \textsc{MAX-3SAT} problem consists in finding a maximal possible number of clauses that can be 
satisfied in $\varphi.$ 
Note that the reduction that has been presented in the previous proof is an $S$-reduction from the \textsc{MAX-3SAT} problem to the $2$\textsc{LFS} problem. 
In the reduction we use the $3CNF$ with exactly three literals in every clause. It has been proved in \cite{Hastad01Some} that \textsc{MAX-3SAT}
is inapproximmable with factor better than $\frac{7}{8}.$ Since our reduction is an $S$-reduction
also $2$\textsc{LFS} cannot have a better approximation. Inapproximability result can easily be generalized to the $k$\textsc{LFS} for $k\leq 2,$
because it is enough to repeat construction with $\sigma^j_s=\sigma^j_s$ for $j>2.$

Analyzing the same reduction and using the same inapproximability result \cite{Hastad01Some}, it can be proven that there is no approximation
for the $2$\textsc{U} problem with an approximation factor better than $\frac{41}{40}$ unless $P=NP.$

\begin{theorem}
\label{LowApp}
There is no polynomial time approximation algorithm for the $2$\textsc{U} problem with approximation factor better than $1 + \frac{1}{40}$ unless $P=NP.$
\end{theorem}

\begin{proof}
	Assume that there is a polynomial time approximation algorithm $\mathcal{A}$ for the $2$\textsc{U} problem with the factor $1 + \frac{1}{40} - \epsilon$ for some $\frac{1}{40} > \epsilon > 0.$ 
	Let $\mathcal{A}(\Gamma(\varphi))$ be a solution for the instance $\langle \Gamma_s, \Gamma_t\rangle$ of the $2$\textsc{U} problem which was created in the proof of Theorem \ref{firstT}. 
	Denote by $|\mathcal{A}(\Gamma(\varphi))|$ the size of the set returned by $\mathcal{A}$ on the input $\Gamma(\varphi).$ 
	Let $2m+k$ be the optimal solution of the $2$\textsc{U} for an instance $\Gamma(\varphi)=\Gamma(\varphi),$ where $m$ is the number of clauses in $\varphi$ which was encoded in $\Gamma(\varphi)$ in the proof of Theorem  \ref{firstT}. 
	$2m+k$ is the size of the smallest feasible solution for $\Gamma(\varphi)$ if and only if $m-k$ is the maximal number of clauses that can be satisfied in $\varphi.$ Let $opt_\varphi=m-k$ We have that
	\begin{eqnarray}
	\label{estimation1}
	|\mathcal{A}(\Gamma(\varphi))| & \leq & (2m+k)\left(1+\frac{1}{40} - \epsilon\right)
	\end{eqnarray}
	Let $\mathcal{A}'(\Gamma(\varphi))$ be an algorithm that executes $\mathcal{A}$ on the input $\Gamma(\varphi)$ and returns $[3m]\setminus Z,$ whereas $\mathcal{A}$ returns $Z$ on the input $\Gamma(\varphi).$ 
	$[3m]\setminus Z$ is a feasible solution for instance $\Gamma(\varphi)$ of the problem $2$\textsc{LFS}
	as well as the $\kappa^{-1}(Z)$ is the satifiablity witness for at least $|[3m]\setminus Z|$ clauses of $\varphi.$ Since $\mathcal{A}(\Gamma(\varphi))=Z$ and 
	by inequality (\ref{estimation1}) we have:
	\begin{eqnarray}
	\label{estimation15}
	|\mathcal{A'}(\Gamma(\varphi))|=|[3m]\setminus Z| & \geq & (m - k) - \left(\frac{1}{40} -\epsilon \right) (2 m + k)
	\end{eqnarray}
	\begin{eqnarray}
	\frac{1}{8} - \frac{2}{40} & = &  \frac{3}{40} \nonumber \\
	\label{estimation2}
	\left( \frac{1}{8} - \frac{2}{40} \right) (m-k) \geq \left( \frac{1}{8} - \frac{2}{40} \right) \left(\frac{1}{2} m \right)& = & \frac{3}{40} \left(\frac{1}{2} m \right) \geq \frac{3}{40} k
	\end{eqnarray}
	If $\varphi$ is a formula in the $3CNF$ form, then there exists an assignment that satisfies at least half of the $\varphi$ clauses. We use this fact to obtain
	inequalities $(\ref{estimation2}).$ Indeed, $m-k$ is the maximal number of clauses of $\varphi$ which can be satisfied, then $m-k\geq 1/2 m$ and 
	$k$ denotes the number of clauses satisfied in the optimal assignment, then $1/2m \geq k$.
	Therefore
	\begin{eqnarray}
	\left( \frac{1}{8} - \frac{2}{40} \right) (m-k)& \geq & \frac{3}{40} k \nonumber\\
	\frac{1}{8} (m-k) & \geq &  \frac{2}{40} (m-k) + \frac{3}{40} k \nonumber \\
	(m-k) - \frac{7}{8} (m-k) & \geq & \frac{1}{40} (2m+k) \nonumber \\
	\label{estimation3}
	(m-k) - \frac{1}{40} (2m+k) & \geq & \frac{7}{8} (m-k)
	\end{eqnarray}
	It is easy to note that
	\begin{eqnarray}
	\label{estimation4}
	\epsilon (2m+k)& \geq & \epsilon (m-k)
	\end{eqnarray}
	Adding inequalities (\ref{estimation3}) and (\ref{estimation4}) by sides we obtain:
	\begin{eqnarray}
	\label{estimation5}
	(m-k) - \frac{1}{40} (2m+k)+\epsilon (2m+k)& \geq & \frac{7}{8} (m-k) + \epsilon (m-k) \nonumber\\
	(m-k) - \left( \frac{1}{40} - \epsilon \right) (2m+k)& \geq & \left(\frac{7}{8} + \epsilon \right)(m-k)=\left(\frac{7}{8} + \epsilon \right)opt_{\varphi}
	\end{eqnarray}
	
	The left side of inequality (\ref{estimation5}) is equal to the right side of inequality (\ref{estimation15}). Thus
	
	\begin{eqnarray*}
	|\mathcal{A'}(\Gamma(\varphi))|=|[3m]\setminus Z| & \geq & \left(\frac{7}{8} + \epsilon \right) opt_{\varphi}
	\end{eqnarray*}
	
	We have shown that $\mathcal{A'}$ is an approximation algorithm for \textsc{Max-3SAT} with an approximation factor $\frac{7}{8} + \epsilon$ 
	for some $ \frac{1}{40} > \epsilon > 0.$ By the result from \cite{Hastad01Some} it is possible only under assumption that $P=NP.$ Hence our assumption 
	about the existence of an approximmation algorithm with a factor $1+\frac{1}{40}$ for the $2$\textsc{U} problem can be true only under the assumption that $P=NP.$
	\end{proof}

By the Lemma \ref{TwoSReductions}, $\textsc{MaxClique} \leq_S n\textsc{LFS}.$ Therefore there exists an $S$-reduction $(f,g)$ between that problems.
We can ask about existence of the solution of the size $l$ for $x' \in \mathcal{I}($n$\textsc{LFS}).$ Let $A_{LFS}$ be an algorithm solving that question for any $x \in \mathcal{I}($n$\textsc{MaxClique}).$ 
One can build an algorithm $A_{MaxC}$ for the \textsc{MaxClique} problem. We do this in the following steps: Count in polynomial time $x'=f(x).$ Ask if there is a solution for instance $x'$ of the size $l.$
Return answer 'yes' if and only if the given anwser is 'yes'. Such an algorithm is correct by definion of S-reduction: existence of a feasible solution $y'$ for $x'$ implies existence existence of feasible solution 
$y=g(x,y')$ with the same size $l.$ This is by condition $cost_{\textsc{MaxClique}}(x,g(x,y'))=cost_{n\textsc{LFS}}(x',y').$ Therefore, $n$\textsc{LFS} is $NP$-hard.
Moreover, by H{\aa}stad and Zuckerman results \cite{Hastad96}, \cite{Zuckerman07}, which establishes that approximation of \textsc{MaxClique} within factor $n^{1-\epsilon}$
is \textsc{NP}-hard, we obtain the following nonapproximability theorem:

\begin{theorem}
\label{rGTheorem}
The approximation of the $n$\textsc{LFS} problem within $n^{1-\epsilon}$ is \textsc{NP}-hard, for any $\epsilon > 0.$
\end{theorem}
Since there is an $S$-reduction from \textsc{MinVertexCover} to $n$\textsc{U}, the lower bound for the approximation factor of \textsc{MinVertexCover} is also
the lower bound for the approximation factor of $n$\textsc{U}. As a conclusion from Khot and Regev result \cite{KhotRegev} we obtain:

\begin{theorem}
	\label{UGC2epsilon}
	There exists no polynomial time $(2-\epsilon)$-approximation algorithm for the $n$\textsc{U} problem unless
	The Unique Game Conjecture is not true.
\end{theorem}

Let $G=(V,E)$ be an undirected graph. The product of graph $G,$ denoted by $G^2=(V^2,E^2)$ is a graph such that $V^2=V\times V$ and $E^2=$ $=\{((u,u'),(v,v'))\;|\; \mbox{ either } u=v \; \wedge\; (u',v')\in E \mbox{ or } (u,v) \in E\}.$
For permutations in $S^k_n,$ in the proof of the next theorem, we develope approach analogous to application of the next lemma in nonapproximality of the \textsc{MaxClique} problem with an approximation factor better than $\sqrt{n}.$

\begin{lemma} \label{folk}
Graph $G$ has a clique of size $k$ if and only if graph $G^2$ has a clique of size $k^2$ \cite{1994papadimitriou}, chapter 13.
\end{lemma}

\begin{theorem}
	\label{sqrtLFS}
	For any constant $c\in \mathbb{N},$ it is \textsc{NP}-hard to approximate the $n^{{1}/{c}}$\textsc{LFS} problem within $n^{1-\epsilon}$ , for any $\epsilon > 0.$   
\end{theorem}
\begin{proof}
	Assume that $n=\nu^{c}.$  Let $\Gamma =\langle (\sigma^{1}_{s},\ldots,\sigma^{\nu}_{s}),$ $(\sigma^{1}_{t},\ldots,\sigma^{\nu}_{t}) \rangle,$ where  
	$\sigma^{i}_{s},\sigma^{i}_{t} \in S_\nu.$  
	For $c > 1$ and $i\in [\nu],$ let $f_i:[\nu^{c-1}] \mapsto [i\nu^{c-1}]\setminus [(i-1)\nu^{c-1}]$ satisfy equation 
	$f_i(k)=(i-1)\nu^{c-1}+k.$ Let $\lambda=(a_1,\ldots,a_\nu) \in S_{\nu^{c-1}}.$ 
	If $[\nu^{c-1}]$ is the domain of $f$ and the function is total over its domain, we write $f(\lambda)$ to denote a tuple 
	$(f(a_1),f(a_2),\ldots,f(a_\nu)).$

	We will denote by $\bigodot$ the generalized concatenation of sequences. If $a=(i_1,i_2,\ldots,i_\alpha) \in S_{\alpha}$ is a
	sequence of indexes, $\lambda_1,\lambda_2, \ldots, \lambda_\alpha$ are finite sequences, then $\bigodot\limits_{k\rightarrow a} \lambda_k$ denotes a concatenation $\lambda_{i_1}\cdot\lambda_{i_2} \cdots \lambda_{i_\alpha}.$  

	We define $\Lambda=T^c=\langle (\lambda^{1}_{s}, \ldots, \lambda^{\nu}_{s} ),$ $(\lambda^{1}_{t}, \ldots, \lambda^{\nu}_{t}) \rangle$
	recursively using $\Gamma$ and $T^{c-1} =\langle (\tau^{1}_{s},\ldots,\tau^{\nu}_{s}), (\tau^{1}_{t},\ldots,\tau^{\nu}_{t}) \rangle,$
	where $\tau^{i}_{s},\tau^{i}_{t} \in S_{\nu^{c-1}}.$ For $c=1,$ $T^c=\Gamma.$
	The $\lambda^{i}_{s},\lambda^{i}_{t} \in S_{\nu^c}$ are defined by the following equations

	\[\lambda^{i}_{s}= \bigodot\limits_{k\rightarrow \sigma^{i}_{s}}f_k(\tau^{i}_{s})\;\;\;\;\;\;\;\;\;\;
	\lambda^{i}_{t}= \bigodot\limits_{k\rightarrow \sigma^{i}_{t}}f_k(\tau^{i}_{t}).\]

	\begin{remark} \label{zwiekszprim}
		Assume that  $\Gamma$  has a feasible solution of size $m$ and $T^{c-1}$ has 
		a solution of size $m^{c-1}.$ Then $\Lambda$ has a solution of size $m^{c}.$
	\end{remark}
	Define auxiliary functions $block^c$ and $inblock^c$ both of type $[n=\nu^c] \mapsto [\nu]:$
	\[block^c(s)=((s-1) \; \div \; \nu^{c})+1, \;\;\;\;\;\;\; inblock^c(s)=((s-1)\mod \nu^{c})+1 .\]
	For a given $J \subseteq [\nu^c]$ let  $W^c(J):=\{l\;|\; \exists r\in J \; l= block^{c-1}(r)\},$ 
	$H^c(l,J):=\{ inblock^{c-1}(r)\;|\; r \in J \mbox{ and } l=block^{c-1}(r) \},$
	$\alpha^c(J)=\max\{|H^c(l,J)| \;| \; l\in [\nu]\}, \beta^c(J):=\max\{l\;|\;|H^c(l,J)|=\alpha^c(J)\}.$

	\begin{remark} \label{zmniejszprim}
		Assume that $c > 1$ and $\Lambda$ has a solution  $J$ of size $m^{c}.$
		Then  $T^{c-1}$ has a solution of the size $m^{c-1}$ and of the form $H^c(\beta^c(J),J),$ or $\Gamma$ has solution of size $m$ and of the form $W^c(J).$
	\end{remark}
	Now we will show, that if there exists a polynomial algorithm solving the $n^{1/c}$\textsc{LFS} problem, with an approximation 
	factor within $O(n^{1-\epsilon})$ for some $0< \epsilon < 1,$ then there exists a polynomial time algorithm that approximates 
	$\nu$\textsc{LFS} with factor $\nu^{1-\gamma}$ for some $0< \gamma < 1.$ 
	Instances of the $\nu$\textsc{LFS} problem consist of permutations from $S_\nu.$ 
	Asssume the existence of a polynomial time algorithm $\cal C$ that approximates $n^{1/c}$\textsc{LFS} with a factor $n^{1-\epsilon}$ 
	and let the polynomial be denoted by $p.$ Approximation algorithm $\mathcal{B}$ for $\nu\textsc{LFS}$ works as follows:

	\begin{enumerate}
		\item Create $\Lambda =\langle (\lambda^{1}_{s},\ldots,\lambda^{\nu}_{s}), (\lambda^{1}_{t},\ldots,\lambda^{\nu}_{t}) \rangle$ as it was defined previously.
		\item Run algorithm $\cal C.$ Assume that it has returned the set $J\subseteq [n].$
		\item Execute recursive procedure $\textsc{ExtractLFS}(\Lambda,c,J).$
		
		\begin{algorithm}[H]
			  \SetKwFunction{FSum}{\textsc{ExtractLFS}}
			  \SetKwProg{Fn}{Procedure}{:}{}
			  \Fn{\FSum{$\Lambda$, $c$, $J$}}{
					\uIf{$\alpha^c(J) < |J|^{1/c} $}{
						\KwRet $W^c(J)$\;
					}
					\Else{
						\uIf{$c=2$}{
							\KwRet $H^c(\beta^c(J),J)$\;
						}
						\uIf{$c>2$}{
							$J=H^c(\beta^c(J),J)$\;
							$\Lambda=T^{c-1}$\;
							\KwRet \FSum{$\Lambda$, $c-1$, $J$}\;
						}
					}
			  }
			\end{algorithm}
	\end{enumerate}

	One can prove, by induction on $c,$ that for input $\Lambda=T^c$ and its solution $J$ the procedure 
	$\textsc{ExtractLFS}(\Lambda,c,J)$ returns a solution of $\Gamma$ of size not less than $|J|^{1/c}.$
	Induction step can be done with the following assumtion.
	
	\begin{assumption} 
	Let $J$ be a solution of $\nu$\textsc{LFS} for instance $T^{c-1}.$ Then \linebreak $\textsc{ExtractLFS}(T^{c-1},c-1,J)$ returns a solution of $\Gamma$ 
	of size not less than $|J|^{1/(c-1)}.$
	\end{assumption}
	
	In the inductive proof we use Remark $\ref{zmniejszprim}$ and the fact that $T^1=\Gamma$ (for the first step of induction). 

	\begin{corollary}  \label{corr}
		If $\Lambda$ has a solution of size $m^c,$ then $\Gamma$ has a solution of size $m.$
	\end{corollary}

	We conclude that, only one among the inequalities, either $|W^c(J)|< |J|^{1/c}$ or $|H^c(\beta^c(J),$ $J)| < |J|^{(c-1)/c}$ can be true. 
	Indeed, if both  inequalities hold there are fewer than $|J|^{1/c}$ intervals of the form $[i\nu^{c-1} + 1, (i+1)\nu^{c-1}]$ that
	intersect with $J$ and each interval which contains some element of $J$ also contains less than $|J|^{(c-1)/c}.$ This means
	that the interval $[1,\nu^c]$ contains fewer than $|J|$ elements. This is contradiction. Thus, by Remark \ref{zmniejszprim} the solution 
	returned by $\textsc{ExtractLFS}(\Lambda,c,J).$ is not smaller than $|J|^{1/c}.$  

	The time complexity of the presented reduction depends on the time complexity of the $T^c$ construction, the complexity of
    $\textsc{ExtractLFS}(T^c,c,J)$ and $\mathcal{C}.$ 
	The complexity of the transformation  $T^{c-1}$ into $T^c$ can be estmiated by 
	$O(\nu^2\cdot \nu^{c-1}).$ The cost of building $f_k(\tau)$ costs $O(\nu^{c-1})$ because it depends on the length
	of $\tau$ element. This cost is multiplied by $O(\nu)$ in the process of the concatenation of $f_k(\tau)$ components 
	and the whole outcome is created $\nu$ times during $\lambda$ components construction. Since the whole construction
	starts from $\Gamma=T^1$ the overall cost is $\Sigma_{i=2}^c O(n^{i+1}).$ This is $O(n^{c+1})$ assuming that $c$ is a
	constant. The complexity of $\textsc{ExtractLFS}(\Lambda,c,J)$ depends on how $\alpha^c(J),$ $W^c(J),$ and $H^c(\beta^c(J),J)$ are determined.
	The rough estimation of the upper bound for those costs is $O(\nu^{2c}).$ 
	The complexity of $\mathcal{C}$ is $O(p(\nu^{c+1})).$ The complexity of all three steps of the reduction 
	is $O(p(\nu^{2c}))=O(p^2(n))$ -- we assume that degree of $p$ is at least $1.$

	Now let estimate the approximation factor. By the assumption about the approximation factor of $\cal C,$ we have

	\[\frac{|Opt_{n^{1/c}\textsc{LFS}}(\Lambda)|}{|\mathcal{C}(\Lambda)|} \leq n^{1-\epsilon}=\nu^{c(1-\epsilon)}\] for any instance $\Lambda.$ The outcome of algorithm $\mathcal{B}$ satisfies 
	\[|\mathcal{B}(\Gamma)| \geq |\mathcal{C}(\Lambda)|^{1/c} = |J|^{1/c}\]
	But also by Remarks \ref{zwiekszprim} and Corollary \ref{corr} the following equation holds 
	\[|Opt_{\nu\textsc{LFS}}(\Gamma)| = |Opt_{n^{1/c}\textsc{LFS}}(\Lambda)|^{1/c}.\] 
	Therefore
	\[\frac{|Opt_{\nu\textsc{LFS}}(\Gamma)|}{|\mathcal{B}(\Gamma)|} \leq  \left( \frac{|Opt_{n^{1/c}\textsc{LFS}}(\Lambda)|}{|\mathcal{C}(\Lambda)|} \right)^{1/c} \leq \left( \nu^{c(1-\epsilon)} \right)^{1/c} = \nu^{1-\epsilon}.\]
	It means that the $\nu\textsc{LFS}$ problem has a polynomial time approximation algorithm with factor $\nu^{1-\epsilon}.$
	This contradicts Theorem \ref{rGTheorem}.
\end{proof}

\begin{theorem}
	For any constant $c<n,$ there exists no polynomial time $(2-\epsilon)$-approximation algorithm for the $n^{1/c}$\textsc{U} problem unless
	The Unique Game Conjecture is not true.
	\end{theorem}
	\begin{proof} (Sketch). Let $n=\nu^c.$ Instances of the $\nu$\textsc{U} problem are of the form $\Gamma = \langle \Gamma_1, \Gamma_2 \rangle,$ where 
		$\Gamma_1,\Gamma_2 \in S^\nu_\nu.$ Consider the following reduction from $\nu$\textsc{U} to $n^{1/c}$\textsc{U}. 
		For a given instance $\Gamma = \langle (\sigma^{1}_{s},\ldots,\sigma^{\nu}_{s}), (\sigma^{1}_{t},\ldots,\sigma^{\nu}_{t}) \rangle$ of $\nu$\textsc{U}, let 
		\[\lambda^{i}_{s} = \bigodot\limits_{k\rightarrow (1,\ldots, \nu^{c-1})}g_k(\sigma^{i}_{s})\;\;\;\;\;\;\;\;\;\;
		\lambda^{i}_{t}= \bigodot\limits_{k\rightarrow (1,\ldots, \nu^{c-1})}g_k(\sigma^{i}_{t}).\]
		In this proof, $g_k(p)=(k-1)\nu + p.$ The $\Lambda=\langle (\lambda^{1}_{s},\ldots,\lambda^{\nu}_{s}), (\lambda^{1}_{t},\ldots,\lambda^{\nu}_{t})\rangle \in S^{\nu}_n$ is an instance of $n^{1/c}$\textsc{U}.
		If there exists algorithm $C,$ that returns a feasible solution $J$ which is smaller than $(2-\epsilon)|Opt_{n^{1/c}\textsc{U}}(\Lambda)|,$ then 
		for $M=\operatorname*{argmin}_{k \in J} |H^2(k,J)|$ and $l\in M,$ $H^2(l,J)$ is a feasible solution of $\Gamma$ and $|H^2(l,J)|\leq (2-\epsilon)|Opt_{\nu\textsc{U}}(\Gamma)|.$ 
	\end{proof}
	
\subsection{On parametrized complexity}

An $S$-reduction $(f,g)$ preserves the costs values between reduced instance $x$ of an $A$ problem and outcome instances $x'$ of a $B$ problem. 
In the case of a minimization problem, the decision version of the problem $B$ is formulated as
the question if $cost_B(x',y') \leq l$ for a given $l \in [n]$, on instance $x'$ of $B,$ $y'$ which is a feasible solution of instance $x'.$
Analogously the decision version of the problem $A$ is the question if $cost_A(x,g(x,y')) \leq l,$ for instance $x$ of $A.$
Assuming that $l$ is parameter, by the equality of costs, the $f$ function is also the reduction $R.$ It is easy to note that 
$x\in A$ iff $x'\in B.$ The $R$ is an fpt-algorithm because the $f$ is polynomially computable. The function $h$ can be assumed to be identity and 
$\kappa_A(x)=cost_A(x,g(x,y')) ,\kappa_B(x)=cost_B(f(x),y'),$ where $y'$ is feasible solution of the problem $B$ for an instance $f(x)$.
The above shows that if $(f,g)$ is $S$-reductions, then $f$ is an fpt-reduction. Similar argument can be applied for maximization
problems.

Therefore, $k\textsc{LFSD} \leq_{FPT} \textsc{MaxClique}$ and $k\textsc{UD} \leq_{FPT} \textsc{VertexCover}.$ Since 
\textsc{MaxClique} is in $W[1],$  \textsc{VertexCover} is in $FPT$ and both parametrized classes $W[1]$ and $FPT$  are closed under fpt-reductions \cite{FlumGrohe06} we obtain:

\begin{theorem}
\label{thirdT}
The k\textsc{LFSD} problem is in $W[1]$ class for any $k\in [n]$.  The k\textsc{UD} problem belongs to the $FPT$ class for any $k\in [n]$.
\end{theorem}

\begin{theorem}
For any constant $c,$ the $n^{{1}/{c}}$\textsc{LFSD} problem is $W[1]$-hard.
\end{theorem}
\begin{proof}(Sketch). Previously we have noticed that if a $(f,g)$ is an $S$-reduction then $f$ is an $fpt$-reduction.
	In the proof of Theorem \ref{TwoSReductions} we show the $S$-reduction from \textsc{Clique} to $n$\textsc{LFSD}.
	Hence there is $fpt$-reduction from \textsc{Clique} parametrized by the size of required clique to $n$\textsc{LFSD} parametrized
	by the required size of set which induces common sequences for all pairs of permutations. Since  the \textsc{Clique} is $W[1]$-hard 
	the parametrized $n$\textsc{LFSD} is $W[1]$-hard under $fpt$-reductions.

	In the proof of Theorem \ref{sqrtLFS} we defined polynomial time reduction which for a given instance $\Gamma$
	of $n$\textsc{LFSD} returns an instance $\Lambda$ of $n^{1/c}$\textsc{LFSD}. Moreover the $\Lambda$ has solution 
	of the size $m^c$ iff the $\Gamma$ has solution of the size $m.$ It can be shown by induction over the $c$ that 
	presented reduction satisfies conditions $(1)$ and $(2)$ of the $fpt$-reduction definition. 
	The condition $(3)$ is satisfied by function $h(y)=y^c,$  for any positive integer constant $c$.
\end{proof}
\section{Initial Motivation and Conclusions}

In order to obtain novelty solution for hypercubes packing we believe that there exists 
some other metaheuristics which allow to find solution in $n$-dimensional space.
This is our motivation to extend the Ulam's metric to $n$-dimension.
We believe that there exists some mataheuristics which use finding a minimal 
distance between $n$-tuples of permutations providing that this way is computationally effective.
For this reason we rely our study on finding a minimal distance between n-tuples of permutations.
Just because \textsc{NP}-hardness of the $k$\textsc{UD} problem, we realise that finding an effective solution seems 
almost imposible. Sill we hope that these metaheuristics can rely on approximation algorithms.
The hope constitutes our motivatin to persue the study in which we proved that there exist
some approximation algorithms for $k$\textsc{U} problem that have the best approximation factors ranged
from to $\frac{41}{40}$ and $2.$ However we need to stress that our findings can not be found as the best 
ones to resolve the problem because we not obtained the strict factors for each $k.$

We show that hierarchies studied in this paper, namely  $k$\textsc{U} and $k$\textsc{LFS}, refer to each other just in a way like
\textsc{MaxClique} and \textsc{MinVertexCover} do. Additionally, we shed light on the question how our hierarchies can be placed 
in the map of parameterized classes of computational complexity.
\bibstyle{plainurl}
\bibliography{biblio2}
\section{Appendix}

\subsection{Proof of Remark \ref{RFour}}
\begin{proof}
We consider two cases $(1)$ if $i\in K$ and $(2)$ if $i\not \in K.$ A \textit{conditional split} 
of $K$ is defined to be a pair of disjoin sets $K^{+}_i$ and $K^{-}_i$ that satisfy the following conditions:
\begin{enumerate}

\item \[
K^{+}_i = \left\{
\begin{array}{ll}
K \cap P^i & \mbox{ if } i \not \in K\\
(K \cap P^i) \cup \{i\} & \mbox{ if } i \in K
\end{array}
\right.
\]

\item $K^{-}_i = K \cap N^i$

\end{enumerate}
We will denote by $\alpha_i$ the non repeating sequence of all elements from $K_i^{-}$ in the increasing order,
by $\beta_i$ the non repeating sequence of all elements from $K_i^{+}$ in the increasing order. It is easily seen that depending on if $i\in K$ or $i \not \in K$ the join sequence $\alpha_i\cdot\beta_i$ has either $|K|-1$ or $|K|$ elements. 

Notice that in case $i\in K$ the sequence $\alpha_i$ is empty and $\beta_i$ consists of elements from $K\setminus \{i\}.$ In this case $\sigma^{i}_{s}$ contains subsequence $(i)\cdot \beta_i.$ In permutation $\sigma^{i}_{s}$ element $i$ appears as
the first element of the sequence. Since $\beta_i$ contains elements from $K_i^+$ and $K_i^+,$ by definition, is a subset of $P^i$ and both $b_i$ and $\beta_i$ have the same order, then $\beta_i$ is subsequence of $b_i.$ 
Permutation  $\sigma^{i}_{t}$ contains subsequence $(i)\cdot b_i$ in this case. Thus $\sigma^{i}_{t}$ also contains subsequence $(i)\cdot \beta_i$. Hence both permutations $\sigma^{i}_{s}$ and $\sigma^{i}_{t}$ contain 
subsequence $(i)\cdot \beta_i$. The common subsequence $(i)\cdot \beta_i$ is $K$-induced subsequence for both $\sigma^{i}_{s}$ and $\sigma^{i}_{t}.$

Consider the case $i\not \in K.$ In this case $\alpha_i$ does not have to be empty sequence, unlike it was previously. 
Since $K_i^{-} \subseteq N^i$ and $K_i^{+} \subseteq P^i,$ the sequence $\alpha_i$ is subsequence of $a_i$ and $\beta_i$ is subsequence of $b_i.$ 
In permutation $\sigma^{i}_{s},$  element $i$ precede $a_i\cdot b_i,$ $\alpha_i\cdot \beta_i$ is subsequence of $a_i \cdot b_i.$ 
Hence $\alpha_i \cdot \beta_i$ is subsequence of $(i)\cdot a_i \cdot b_i= \sigma^{i}_{s}.$ By definition $\sigma_{t}^i = a_i \cdot (i) \cdot b_i.$ 
Hence $\alpha_i \cdot \beta_i$ is subsequence of $\sigma_{t}^i.$  Thus $\alpha_i\cdot \beta_i$ is $K$-induced subsequence for $\sigma^{i}_{s}$ and $\sigma^{i}_{t}.$ 
Now, by definition, $K$ is a feasible solution of $\langle \Gamma_s, \Gamma_t \rangle.$

Now assume that a set $K\subseteq [n]$ induces common susequences for all $(\sigma^i_s,\sigma^i_t),$ where $i\in [n].$ Let $\gamma_i$ be a subsequence of $\sigma^i_s=(i)\cdot a_i\cdot b_i$ and $\sigma^i_t=a_i \cdot (i) \cdot b_i.$ 
Let $i\in K.$ It is easy to see that $i$ is the first element of $\gamma_i,$ hence $\gamma_i=(i)\cdot \gamma'_i,$ where $\gamma'_i,$ is the subsequence of $b_i.$ By definition of $b_i$ all its elements are in $P_i.$
This means that $iEx$ for all $x\in \{\gamma'_i\}.$ Moreover $\{\gamma'_i\}=K\setminus \{i\}.$ Now we can make a conclusion that for all different $i,x \in K$  $iEx.$ This means that $K$ is a clique in $G.$
\end{proof}

\subsection{Proof of Remark \ref{RFive}}
\begin{proof}
Suppose that some $B \subseteq V,$ which satisfies $|B|>m,$ induces subsequences for all 
$(\sigma^{i}_{s},\sigma^{i}_{t}).$ Subgraph of $G$ induced by $B$ cannot be a clique. 
Therefore, there are $\alpha, \beta \in B,$ such that  $\neg \alpha E \beta.$  In permutation $\sigma_{s}^\alpha$ element $\alpha$ precedes $\beta$ because, by definition of $\sigma_{s}^\alpha,$ $\alpha$ precedes all other members of $V^{-\alpha}.$ In permutation $\sigma_{t}^\alpha$ elements from $a_\alpha$ precedes $\alpha.$  By definitions of $U^{\alpha}$ and $a_\alpha$ we have $\beta$ occurs in $a_\alpha.$ Thus 
$\alpha$ and $\beta$ appear in $\sigma_{t}^\alpha$ and $\sigma_{s}^\alpha$ in reverse order. Therefore, there is no sequence induced by $B$ which is common subsequence of $\sigma_{t}^\alpha$ and $\sigma_{s}^\alpha.$ This is contradiction. 
\end{proof}

\subsection{Proof of Remark \ref{zwiekszprim}}

\begin{proof}  Let $B^{c-1} \subseteq [\nu^{c-1}]$ be a set, which induces common subsequence for all $(\tau^{x}_{s},\tau^{x}_{t}).$ 
Let $B \subseteq [\nu]$ be a set, which induces common subsequence for all $(\sigma^{x}_{s},\sigma^{x}_{t}).$
Any pair of different elements $i,j \in B$ appears in the same order in $\sigma^{x}_{s}$ and $\sigma^{x}_{t}$ for all $x \in [\nu].$ 
This means that $i$ occurs in the earlier position than $j$ in $\sigma^{x}_{t}$ if and only if $i$ occurs in the earlier position than $j$ in $\sigma^{x}_{s}.$
Similar remark holds for all pairs $(\tau^{x}_{s},\tau^{x}_{t}) \in S^2_{\nu^{c-1}}$ and $k,l \in [\nu^{c-1}].$

We claim that $B^c=\{(j-1)\nu^{c-1}+k\;|\; j \in B, k\in B^{c-1}\}$ is solution of $\Lambda.$ 
On the contrary, suppose that $B^c$ do not induces common subsequence of all pairs $(\lambda^{x}_{s},\lambda^{x}_{t}).$ 
Thus we may find $y$ such that $B^c$-induced sequences $\rho$ and $\gamma,$ for $\lambda^{y}_{s}$ and $\lambda^{y}_{t},$ 
and $\rho \neq \gamma,$ respectively. Then there are elements $(i-1)\nu^{c-1}+l$ and $(j-1)\nu^{c-1}+k$ that cause first difference in $\rho$ and $\gamma.$ 
Thus either $i \neq j$ or $i=j.$ For the first case, by definition of $\lambda^{y}_{s}$ and $\lambda^{y}_{t},$ without loss of generality, one can assume that 
in the  sequence $\bigodot\limits_{r\rightarrow \sigma^y_s} f_r(\tau^y_s),$ element $(i-1)\nu^{c-1}+l$ appears before element $(j-1)\nu^{c-1}+k$ and 
in the $\bigodot\limits_{r\rightarrow \sigma^y_t} f_r(\tau^y_t)$ they are in the opposite order. This holds for any $l,k\in [\nu]$ because all elements from
$f_i([\nu^{c-1}])$ appear before elements from $f_j([\nu^{c-1}])$ if $i$ appears before $j$ in $\sigma_s^y.$ Elements $i,j$ are in the opposite order in $\sigma^y_t.$ 
By definition of $B^c;$ $i,j \in B$ and we assumed at the beginning that $B$ induces common subsequence in all $(\sigma^x_s,\sigma^x_t),$ in particular in $(\sigma^y_s,\sigma^y_t).$
This is contradiction.

In the case $i=j,$ element $(i-1)\nu^{c-1}+l$ appears before $(i-1)\nu^{c-1}+k$ in $f_i(\tau^y_s)$ and  in the $f_i(\tau^y_t)$ they are in the opposite order.
This means that $l,k$ are in different orders in $\tau^y_s$ and $\tau^y_t.$ By definition of $B^c,$ $k,l\in B^{c-1},$ but it is assumed $B^{c-1}$
induces common subsequence for all $(\tau^x_s,\tau^x_t),$ in particular in $(\tau^y_s,\tau^y_t).$
This is contradiction.
\end{proof}

\subsection{Proof of Remark \ref{zmniejszprim}}

\begin{proof} Assume that there is $J \subseteq [\nu^c]$ of the size $m^c$ that is solution of $\Lambda.$ 

Let $R(l)=\{s \in J\;|\; inblock^{c-1}(s) \in H^c(l,J)\}.$ First, notice that $J=\bigcup\limits_{l\in W^c(J)} R(l)$ and $R(l)$ are pairwise disjoint sets.
It is easy to see if $|W^c(J)| < m,$ then there exists $l$ such that $|R(l)| \geq m^{c-1}.$ At least one of two cases holds: $W^c(J)$ is a set which is feasible 
solution of $\Gamma$ or $H^c(l,J)$ is feasible solution of $T^{c-1}.$ In particular if $l$ be chosen to make $H^c(l,J)$ of maximal size then $l$ be equal $\beta^c(J).$

For any $p,r \in W^c(J)$ there are $\pi=(p-1)\nu^{c-1}+ k, \rho=(r-1)\nu^{c-1}+ l \in J.$ By definition of $\lambda^i_s,\lambda^i_t,$ if $\pi$ and $\rho$ appear in the same order
in $\lambda^i_s,\lambda^t_s$ then $p$ and $r$ are in the same order in $\sigma^i_s$ and $\sigma^i_t.$ Thus, indeed $W^c(J)$ is a solution of $\Gamma.$

For any $a,b \in H^c(l,J),$ there are $\pi,\rho \in J$ such that $a=inblock^{c-1}(\pi),$ $b=inblock^{c-1}(\rho),$ $l=block^{c-1}(\pi)$ and 
$l=block^{c-1}(\rho).$  By definition of $block^{c-1}$ and $inblock^{c-1},$  $\pi=(l-1)\nu^{c-1} + a$  and 
$\rho=(l-1)\nu^{c-1} + b.$ 

For all $i\in [\nu],$ numbers $\pi$ and $\rho$ occur in the same order in $\lambda^i_s$ and $\lambda^i_t.$ 
Sequences $f_{l}(\tau^i_s)$ and $f_{l}(\tau^i_t)$  are substrings of $\lambda^i_s$ and $\lambda^i_t,$ respectively.
Moreover $\pi$ and $\rho$ occur in $f_{l}(\tau^i_s)$ and $f_{l}(\tau^i_t)$ in the same order. 
It means that, for each $i\in [\nu],$  $a$ and $b$ appear in the same order in $\tau^i_s$ and $\tau^i_t$ (by definition of $f_l$).
Thus, if $a,b \in H^c(l,J),$ then $a$ and $b$ appear in the same order in $\tau^i_s$ and $\tau^i_t.$ Therefore, for all $l\in [\nu],$
$H^c(l,J)$ is feasible solution of $T^{c-1}.$ In case $W^c(J) < m,$  $R(\beta^c(J)) \geq m^{c-1},$ this implies that 
$H^c(\beta^c(J),J) \geq m^{c-1}.$
\end{proof}

\end{document}